%% file: main.tex
\documentclass[aps,pra,reprint,amsmath,amssymb,nofootinbib,longbibliography]{revtex4-2}

\usepackage[T1]{fontenc}
\usepackage{lmodern}
\usepackage{microtype}
\usepackage{amsthm,mathtools,bm}
\usepackage{aliascnt}
\usepackage{booktabs,array,enumitem,graphicx}
\usepackage[hidelinks]{hyperref}
\usepackage[nameinlink,capitalize]{cleveref}

\hypersetup{
  pdftitle={Ancilla-Depth Phase Diagrams for Quantum Reference-Frame Comparison},
  pdfauthor={Maxim V. Churilov},
  pdfsubject={Ancilla-restricted quantum statistical comparison and positive maps},
  pdfkeywords={quantum reference frames, channel comparison, k-positive maps, depolarizing channels, channel deficiency}
}

\newtheorem{theorem}{Theorem}[section]
\newaliascnt{lemma}{theorem}
\newtheorem{lemma}[lemma]{Lemma}
\aliascntresetthe{lemma}
\newaliascnt{proposition}{theorem}
\newtheorem{proposition}[proposition]{Proposition}
\aliascntresetthe{proposition}
\newaliascnt{corollary}{theorem}
\newtheorem{corollary}[corollary]{Corollary}
\aliascntresetthe{corollary}
\newaliascnt{definition}{theorem}
\newtheorem{definition}[definition]{Definition}
\aliascntresetthe{definition}
\newaliascnt{remark}{theorem}

\aliascntresetthe{remark}
\newaliascnt{example}{theorem}
\newtheorem{example}[example]{Example}
\aliascntresetthe{example}

\crefname{theorem}{theorem}{theorems}
\crefname{lemma}{lemma}{lemmas}
\crefname{proposition}{proposition}{propositions}
\crefname{corollary}{corollary}{corollaries}
\crefname{definition}{definition}{definitions}
\crefname{remark}{remark}{remarks}
\crefname{example}{example}{examples}

\newcommand{\M}{\mathsf M}
\newcommand{\CPTP}{\mathsf{CPTP}}
\newcommand{\id}{\mathrm{id}}
\newcommand{\Tr}{\operatorname{Tr}}
\newcommand{\D}{\mathcal D}
\newcommand{\Rmap}{\mathcal R}
\newcommand{\ketbra}[2]{|#1\rangle\!\langle#2|}
\newcommand{\dist}{\operatorname{dist}}

\input{physics_format}

\begin{document}

\title{Ancilla-Depth Phase Diagrams for Quantum Reference-Frame Comparison}
\author{Maxim V. Churilov}
\email{churilovm1305@gmail.com}
\affiliation{Independent Researcher, Orenburg, Russia}
\date{May 8, 2026}

\begin{abstract}
Comparing two noisy quantum reference frames as statistical experiments depends on the dimension of the ancillary memory available to the decision procedure.  For finite-dimensional channels $A$ and $B$ with invertible $A$, we show that exact simulation of all measurements assisted by an $r$-dimensional ancilla is equivalent to $r$-positivity of the unique factor $\Gamma=BA^{-1}$.  The hierarchy can be realized by physical channel pairs: every unital, trace-preserving map that is $k$-positive but not $(k+1)$-positive embeds as the factor between $\D_a$ and $\Gamma\circ\D_a$ on an exact interval determined by the smallest Choi eigenvalue.  For depolarizing source and target channels $\D_a$ and $\D_b$, including negative and singular source parameters, the phase boundary is
\[
 \D_a\succeq_r\D_b
 \quad\Longleftrightarrow\quad
 -\frac{1}{dr-1}\le \frac{b}{a}\le 1
 \qquad (a\ne0).
\]
We derive closed formulas for an operational state-restricted level-$r$ deficiency and for the distance to every physical post-processing,
\[
 \delta_{\rm phys}(\D_b\mid\D_a)
 =\left(1-\frac{1}{d^2}\right)\dist(b,I_a),
 \qquad
 I_a=\operatorname{conv}\!\left\{a,-\frac{a}{d^2-1}\right\}.
\]
The largest physical conversion cost hidden from all tests using ancillas of dimension at most $k$ is $(d-k)/[d(d^2-1)]$.  An untouched $m$-level spectator changes the first detecting external level from $k+1$ to $\lfloor k/m\rfloor+1$.  A transpose--depolarizing construction shows that the separation is not confined to depolarizing factors.  The results quantify the distinction between ancilla-restricted statistical simulation and implementation by a single quantum channel.
For noninvertible source channels we further give an attained semidefinite
factorization deficiency, with an exact Choi-kernel witness for positive
diamond-norm distance from post-processability.  At every fixed number
of adaptive uses, exact conversion and the attained strategy-norm
deficiency are likewise finite semidefinite programmes over deterministic
comb Choi operators.
\end{abstract}

\maketitle

\section{Introduction}

Finite quantum reference frames provide operational descriptions relative to imperfect physical systems rather than ideal classical coordinates \cite{Bartlett2007,GourSpekkens2008,MarvianSpekkens2014}.  After inaccessible degrees of freedom are discarded, a noisy frame can be represented by a channel from a neutral description to the records available in a laboratory.  Comparing two such frames is then a comparison of quantum statistical experiments.  The answer depends on the decision class: an unentangled probe tests positivity of a statistical factor, whereas sufficiently large quantum memory forces complete positivity.

The endpoints are familiar.  Positive maps reproduce unassisted measurement statistics, and completely positive maps are physical post-processings.  Between them lies the hierarchy of $r$-positive maps.  What is less immediate is whether every strict level can occur between two genuine perspective channels, whether the hierarchy admits a closed physical phase diagram, and how far an ancilla-restricted statistical simulation can remain from any physical converter.

For invertible finite-dimensional perspectives these questions have a sharp answer.  Invertibility removes quotient and extension ambiguities and uniquely determines the source-to-target factor as
\begin{equation}
 \Gamma_{B|A}=B\circ A^{-1}.
 \label{eq:factor-intro}
\end{equation}
Exact simulation of every measurement assisted by an $r$-level memory is equivalent to $r$-positivity of this map.  This observation is related to statistical morphisms and restricted randomization criteria \cite{Buscemi2012,Jencova2012,Jencova2016}; our focus is the physical channel pairs, exact phase boundaries, conversion cost, and activation law that follow from it.

The solvable model is isotropic Weyl misalignment.  A depolarizing channel
\begin{equation}
 \D_a(X)=aX+(1-a)\frac{\Tr X}{d}I
 \label{eq:depol}
\end{equation}
is a random-unitary average over discrete phase-space displacements.
Composition multiplies parameters.  For every nonzero physical source
parameter, the unique factor between $\D_a$ and $\D_b$ is
$\D_{b/a}$.  We derive its exact $r$-positivity interval and margin,
including negative source parameters, and treat the singular
completely depolarizing source separately.

A second distinction is essential.  Ancilla-restricted measurement simulation allows the source measurement to depend on the target measurement; it need not arise from one laboratory channel.  We therefore compute the closest physical post-processing separately.  Twirling proves that the optimum is itself depolarizing and gives a closed diamond-deficiency formula.  This exposes a finite gap between statistical simulability and physical convertibility.

\subsection{Relation to prior work and claim boundary}

Quantum Blackwell comparison, statistical morphisms, and channel deficiency are established \cite{Blackwell1953,Shmaya2005,Buscemi2012,Jencova2012,Jencova2016}.  The hierarchy of $k$-positive maps and Schmidt-number witnesses is likewise mature \cite{Choi1975,TerhalHorodecki2000,Stormer2013,VomEnde2025}.  Ancilla dimension and Schmidt-rank resources in channel discrimination have been studied independently \cite{PuzzuoliWatrous2017,CaiaffaPiani2018,BaeEtAl2019}.

The restricted randomization lemma is an adapted finite-dimensional
statement, not a new general randomization theorem.  Likewise, the
matrix-positivity criteria used below are standard.  The contribution is
the combined channel-level analysis: an exact physical embedding radius,
a complete signed-source depolarizing phase diagram, closed state-restricted and
CPTP deficiency formulas, an extremal hidden-conversion gap, and a spectator
activation law for explicit physical random-unitary frame pairs.  These
claims are confined to the finite-dimensional models and decision classes
stated below.

\section{Ancilla-restricted comparison}

Let $A,B:\M_d\to\M_d$ be channels.  For $r\ge1$, write $A\succeq_rB$ if, for every POVM $\{M_y\}$ on the target output tensored with $\mathbb C^r$, there exists a source POVM $\{N_y\}$ such that
\begin{equation}
 \Tr[M_y(B\otimes\id_r)(\rho)]
 =\Tr[N_y(A\otimes\id_r)(\rho)]
 \label{eq:measurement-simulation}
\end{equation}
for every input state $\rho$ and every outcome $y$.

\begin{theorem}[Restricted randomization with an invertible source]
\label{thm:restricted}
Assume $A$ is invertible as a linear map and set $\Gamma=BA^{-1}$.  Then
\begin{equation}
 A\succeq_rB
 \quad\Longleftrightarrow\quad
 \Gamma\text{ is $r$-positive}.
 \label{eq:restricted}
\end{equation}
The factor is trace preserving.  Simulation at all matrix levels is equivalent to $\Gamma$ being a channel.
\end{theorem}

\begin{proof}
If $\Gamma$ is $r$-positive, then $\Gamma^*$ is unital and $r$-positive.  Pulling a target POVM back by $\Gamma^*\otimes\id_r$ gives a source POVM and proves \cref{eq:measurement-simulation}.

Conversely, apply the simulation condition to every two-outcome POVM $\{M,I-M\}$.  Equality for all input states gives
\begin{equation}
 (A^*\otimes\id_r)(N)
 =(B^*\otimes\id_r)(M)
 =(A^*\otimes\id_r)(\Gamma^*\otimes\id_r)(M).
\end{equation}
Invertibility of $A^*\otimes\id_r$ forces $N=(\Gamma^*\otimes\id_r)(M)\ge0$.  Scaling arbitrary positive operators into the effect interval proves $r$-positivity.  Trace preservation follows from surjectivity of $A$ and trace preservation of $A,B$.
\end{proof}

\begin{definition}[Complete-positivity depth]
For invertible $A$, define
\begin{equation}
 \ell(B|A)=\min\{r\ge1:A\not\succeq_rB\},
 \label{eq:depth}
\end{equation}
with $\ell=\infty$ if the factor is completely positive.
\end{definition}

This invariant records the least ancillary dimension required to falsify statistical simulability.  It is different from diamond deficiency, which asks for one physical converter.

\section{Universal physical embedding}

The hierarchy is not confined to unphysical examples.

\begin{theorem}[Physical embedding of every strict level]
\label{thm:embedding}
Let $\Gamma:\M_d\to\M_d$ be unital, trace preserving, $k$-positive,
and not $(k+1)$-positive.  Then, for every
$0<a\le a_{\max}(\Gamma)$,
\begin{equation}
 A_a=\D_a,
 \qquad
 B_a=\Gamma\circ\D_a
 \label{eq:embedding}
\end{equation}
are channels and $\ell(B_a|A_a)=k+1$.  With unnormalized Choi
matrices, the exact endpoint is
\begin{equation}
 0<a\le a_{\max}(\Gamma):=
 \frac1{1-d\,\lambda_{\min}(J(\Gamma))}.
 \label{eq:embedding-radius}
\end{equation}
This interval is exact: $B_a$ is completely positive if and only if
$0\le a\le a_{\max}(\Gamma)$.  A directly norm-computable
conservative condition is
\begin{equation}
 0<a<
 \frac1{d\|J(\Gamma)-I_{d^2}/d\|_\infty}.
 \label{eq:embedding-bound}
\end{equation}
\end{theorem}

\begin{proof}
Because $(k+1)$-positivity fails, $\Gamma$ is not completely positive
and $\gamma_{\min}:=\lambda_{\min}(J(\Gamma))<0$.  At $a=0$,
unitality gives $B_0=\D_0$, whose Choi matrix $I_{d^2}/d$ is positive
definite.  Moreover,
\begin{equation}
 J(B_a)=\frac1dI_{d^2}
 +a\left(J(\Gamma)-\frac1dI_{d^2}\right).
\end{equation}
The identity commutes with $J(\Gamma)$, so the smallest eigenvalue is
\begin{equation}
 \lambda_{\min}(J(B_a))
 =\frac1d+a\left(\gamma_{\min}-\frac1d\right).
\end{equation}
It is nonnegative exactly for
$a\le[1-d\gamma_{\min}]^{-1}$, proving the exact radius.  The norm
condition follows by bounding the perturbation of $I/d$.  Since
$A_a$ is invertible for $a>0$ and
$B_aA_a^{-1}=\Gamma$, \cref{thm:restricted} gives the depth.
\end{proof}

The theorem regularizes the accessible directions, not the factor.  The source and target are physical, while the unique relation between them remains non-CP.

For the generalized reduction factor
$\Rmap_k=\D_{-1/(kd-1)}$, the two Choi eigenvalues are
\begin{equation}
 \lambda_\parallel=-\frac{d-k}{kd-1},
 \qquad
 \lambda_\perp=\frac{k}{kd-1},
 \label{eq:reduction-choi-spectrum}
\end{equation}
with multiplicities one and $d^2-1$, respectively.  Hence
\cref{eq:embedding-radius} gives
\begin{equation}
 a_{\max}(\Rmap_k)=\frac{kd-1}{d^2-1},
\end{equation}
anticipating the sharp physical window obtained below from the
depolarizing phase diagram.

\begin{theorem}[A nondepolarizing transpose wedge at exact ancilla depth two]
\label{thm:transpose-wedge}
For $d\geq2$ define the transpose--depolarizing map
\begin{equation}
 \Theta_\lambda(X)=\lambda X^T+
 (1-\lambda)\operatorname{Tr}(X)\frac{I_d}{d}.
 \label{eq:transpose-depolarizing}
\end{equation}
It is positive exactly for
$-1/(d-1)\leq\lambda\leq1$, whereas for every $r\geq2$ it is
$r$-positive exactly on the completely positive interval
\begin{equation}
 -\frac1{d-1}\leq\lambda\leq\frac1{d+1}.
 \label{eq:transpose-r-positive}
\end{equation}
Consequently, for
\begin{equation}
 \frac1{d+1}<\lambda\leq1,
 \qquad
 0<a\leq\frac1{(d+1)\lambda},
 \label{eq:transpose-physical-wedge}
\end{equation}
the maps
\[
 A_a=\mathcal D_a,
 \qquad
 B_{a,\lambda}=\Theta_\lambda\circ\mathcal D_a
 =\Theta_{a\lambda}
\]
are quantum channels, but their unique factor is $\Theta_\lambda$.
Therefore $A_a\succeq_1B_{a,\lambda}$ while
$A_a\not\succeq_2B_{a,\lambda}$, and the ancilla depth is exactly two
throughout the stated wedge \eqref{eq:transpose-physical-wedge}.
\end{theorem}

\begin{proof}
The unnormalized Choi matrix is
\[
 J(\Theta_\lambda)=\lambda F+\frac{1-\lambda}{d}I_{d^2},
\]
where $F$ is the swap.  Positivity of the map is equivalent to positivity of
$\Theta_\lambda(|\psi\rangle\langle\psi|)$ for every unit vector, whose two
eigenvalues are
$[1+(d-1)\lambda]/d$ and $(1-\lambda)/d$.  This gives the stated positive
interval.  On the symmetric and antisymmetric subspaces, the Choi
eigenvalues are respectively
$[1+(d-1)\lambda]/d$ and $[1-(d+1)\lambda]/d$, giving the CP interval.  If
$\lambda>1/(d+1)$, an antisymmetric vector has Schmidt rank two and negative
Choi expectation.  Thus the map is not $2$-positive, and hence not
$r$-positive for any $r\geq2$; the converse follows from complete
positivity.

Composition with $\mathcal D_a$ gives $\Theta_{a\lambda}$ directly.
Condition \eqref{eq:transpose-physical-wedge} places $a\lambda$ in the CP
interval, so both source and target are channels.  Since $a>0$,
$\mathcal D_a$ is invertible and the factor is uniquely
$B_{a,\lambda}A_a^{-1}=\Theta_\lambda$.  Its positivity and failure of
$2$-positivity give the exact depth claim.
\end{proof}

\section{Exact isotropic phase diagram}

The depolarizing family is closed under composition:
\begin{equation}
 \D_x\circ\D_y=\D_{xy}.
 \label{eq:composition}
\end{equation}
Its matrix positivity can be solved exactly.

\begin{lemma}[$r$-positivity of a depolarizing map]
\label{lem:r-positive}
For $1\le r\le d$,
\begin{equation}
 \D_\lambda\text{ is $r$-positive}
 \quad\Longleftrightarrow\quad
 -\frac1{dr-1}\le\lambda\le1.
 \label{eq:r-positive}
\end{equation}
More precisely, the minimum Choi expectation on normalized vectors of
Schmidt rank at most $r$ is
\begin{equation}
 \mu_r(\lambda):=
 \min_{\substack{\|v\|=1\\\operatorname{SR}(v)\le r}}
 \langle v|J(\D_\lambda)|v\rangle
 =
 \begin{cases}
 \displaystyle\frac{1+\lambda(dr-1)}{d},&\lambda\le0,\\[6pt]
 \displaystyle\frac{1-\lambda}{d},&\lambda\ge0.
 \end{cases}
 \label{eq:block-margin}
\end{equation}
\end{lemma}

\begin{proof}
The unnormalized Choi matrix is
\begin{equation}
 J(\D_\lambda)=\lambda\ketbra{\Omega_d}{\Omega_d}
 +\frac{1-\lambda}{d}I_{d^2}.
 \label{eq:depol-choi}
\end{equation}
For $\lambda<0$, its minimum expectation over unit vectors of Schmidt rank at most $r$ occurs at a maximally entangled vector on an $r$-dimensional subspace, because
\begin{equation}
 |\langle\Omega_d|v\rangle|^2\le r.
\end{equation}
The minimum is $[1+\lambda(dr-1)]/d$.  For $\lambda\ge0$, the positive
overlap term is minimized at zero, which is attained already by a
product vector orthogonal to $|\Omega_d\rangle$ when $d\ge2$.  The
minimum is then $(1-\lambda)/d$.  Its sign gives the upper boundary
$\lambda\le1$, while the negative branch gives the lower boundary in
\cref{eq:r-positive}.
\end{proof}

\begin{theorem}[Depolarizing ancilla-depth phase diagram]
\label{thm:phase}
Let $A=\D_a$ and $B=\D_b$, where $a\ne0$ and
\begin{equation}
 a,b\in\left[-\frac1{d^2-1},1\right].
\end{equation}
Then
\begin{equation}
 A\succeq_rB
 \quad\Longleftrightarrow\quad
 -\frac1{dr-1}\le\frac ba\le1.
 \label{eq:phase}
\end{equation}
In particular, for $1\le k<d$,
\begin{equation}
 -\frac1{dk-1}\le\frac ba<-\frac1{d(k+1)-1}
 \label{eq:depth-band}
\end{equation}
if and only if $\ell(B|A)=k+1$.
Writing $\lambda=b/a$, the complete depth function is
\begin{align}
 \ell(B|A)=1
 &\Longleftrightarrow
 \lambda<-\frac1{d-1}\ \text{or}\ \lambda>1,\\
 \ell(B|A)=k+1
 &\Longleftrightarrow
 -\frac1{dk-1}\le\lambda< -\frac1{d(k+1)-1},\\
 \ell(B|A)=\infty
 &\Longleftrightarrow
 -\frac1{d^2-1}\le\lambda\le1,
 \label{eq:complete-depth}
\end{align}
where the middle line ranges over $1\le k<d$.
\end{theorem}

\begin{proof}
The source is invertible and \cref{eq:composition} gives the unique factor $BA^{-1}=\D_{b/a}$.  Apply \cref{thm:restricted,lem:r-positive}.
\end{proof}

\begin{proposition}[Singular isotropic source]
\label{prop:singular-source}
For $A=\D_0$ and any physical target $B=\D_b$,
\begin{equation}
 \D_0\succeq_r\D_b
 \quad\Longleftrightarrow\quad b=0
 \qquad(r\ge1).
 \label{eq:singular-comparison}
\end{equation}
Moreover,
\begin{equation}
 \delta_{\rm phys}(\D_b|\D_0)
 =|b|\left(1-\frac1{d^2}\right).
 \label{eq:singular-gap}
\end{equation}
\end{proposition}

\begin{proof}
The source output is $\Tr(X)I/d$ and therefore contains no dependence
on the system input.  If $b\ne0$, there is a two-outcome target measurement whose
outcome probabilities differ for two orthogonal pure inputs.  These
input-dependent statistics cannot be simulated from $\D_0$; the
comparison therefore already fails at $r=1$.  If $b=0$, source and target are identical.  For physical conversion, any
channel following $\D_0$ is a replacer.  Twirling the converter cannot
increase the error to the unitarily covariant target, and produces
$\D_0$.  \Cref{eq:depol-distance} then gives
\cref{eq:singular-gap}.
\end{proof}

\begin{proposition}[Nested comparison intervals]
\label{prop:nested-intervals}
For $a\ne0$, define
\begin{equation}
 I_a^{(r)}
 =\operatorname{conv}\left\{a,-\frac a{dr-1}\right\},
 \qquad 1\le r\le d.
 \label{eq:r-interval}
\end{equation}
Then, for physical $b$,
\begin{equation}
 \D_a\succeq_r\D_b
 \quad\Longleftrightarrow\quad b\in I_a^{(r)}.
 \label{eq:interval-order}
\end{equation}
The intervals are nested,
\begin{equation}
 I_a^{(1)}\supseteq I_a^{(2)}
 \supseteq\cdots\supseteq I_a^{(d)}=:I_a,
\end{equation}
and $I_a$ is exactly the CPTP post-processing interval.  Consequently,
physical convertibility is antisymmetric on the depolarizing line:
if $\D_a$ converts to $\D_b$ and $\D_b$ converts to $\D_a$, then
$a=b$.
\end{proposition}

\begin{proof}
Multiplying the ratio interval in \cref{eq:phase} by $a$ gives
\cref{eq:r-interval}, with the convex hull taking care of the sign of
$a$.  The lower ratio boundary increases monotonically with $r$, which
proves nesting, and the $r=d$ interval is
\cref{eq:reachable-interval}.  If $a,b\ne0$ are mutually convertible,
then $c=b/a$ and $c^{-1}$ both lie in
$[-1/(d^2-1),1]$.  Positive $c$ forces $c=c^{-1}=1$; negative $c$
would force one of $|c|,|c^{-1}|$ to exceed one.  The cases involving
zero follow from \cref{prop:singular-source}.
\end{proof}

\begin{figure}[t]
\centering
\includegraphics[width=\columnwidth]{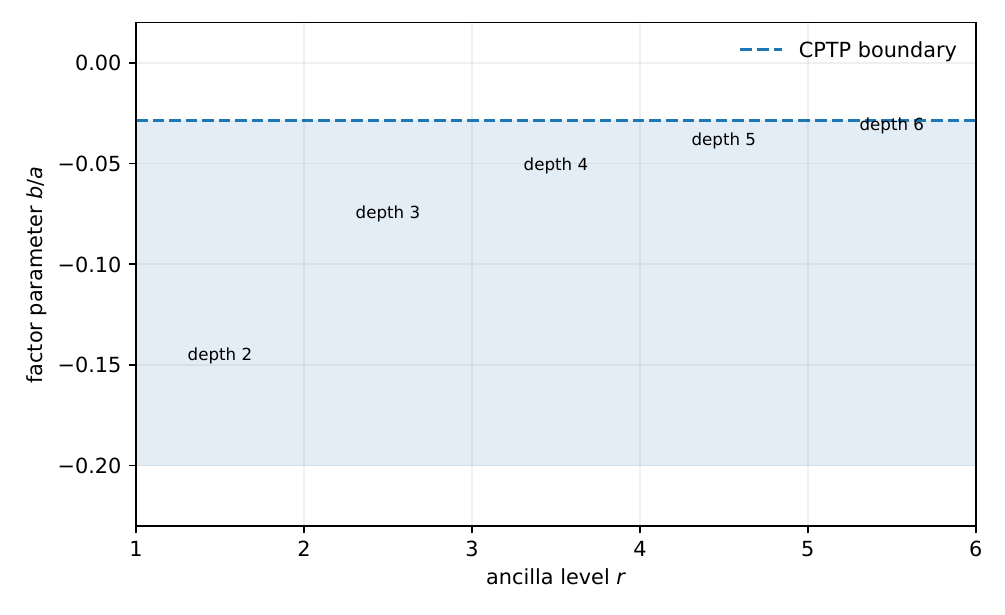}
\caption{Exact positivity bands for the ratio $b/a$ in dimension $d=6$.  Crossing the boundary $-1/(dr-1)$ activates an $r$-dimensional witness.}
\label{fig:phase}
\end{figure}

The boundary factors are normalized generalized reduction maps,
\begin{equation}
 \Rmap_k(X)=\frac{k\Tr(X)I-X}{kd-1}=\D_{-1/(kd-1)}(X).
 \label{eq:reduction}
\end{equation}
They are $k$-positive and not $(k+1)$-positive.  The Schmidt-rank-$(k+1)$ vector
\begin{equation}
 |\omega_{k+1}\rangle
 =\frac1{\sqrt{k+1}}\sum_{j=1}^{k+1}|j,j\rangle
\end{equation}
has Choi expectation
\begin{equation}
 \langle\omega_{k+1}|J(\Rmap_k)|\omega_{k+1}\rangle
 =-\frac1{kd-1}.
 \label{eq:witness-margin}
\end{equation}

\begin{proposition}[Explicit two-outcome decision witness]
\label{prop:explicit-witness}
Let the factor be $\D_\lambda$ and set
\begin{equation}
 M_r=\ketbra{\omega_r}{\omega_r},
 \qquad
 |\omega_r\rangle=\frac1{\sqrt r}\sum_{j=1}^r|j,j\rangle.
\end{equation}
The target two-outcome POVM $\{M_r,I-M_r\}$ can be simulated from an
invertible depolarizing source only if
\begin{equation}
 \langle\omega_r|
 (\D_\lambda^*\otimes\id_r)(M_r)
 |\omega_r\rangle
 =
 \frac{1+\lambda(dr-1)}{dr}\ge0.
 \label{eq:decision-witness-margin}
\end{equation}
Hence $M_r$ is an explicit rejecting decision problem whenever
$\lambda<-1/(dr-1)$.  For the boundary factor $\Rmap_k$ and the first
detecting ancilla $r=k+1$, the negative effect margin is
\begin{equation}
 -\frac1{(k+1)(kd-1)}.
 \label{eq:decision-boundary-margin}
\end{equation}
\end{proposition}

\begin{proof}
The depolarizing map is Hilbert--Schmidt self-adjoint.  Its action on
the rank-one effect is
\begin{equation}
 (\D_\lambda\otimes\id_r)(M_r)
 =\lambda M_r
 +(1-\lambda)\frac{I_d}{d}\otimes\frac{P_r}{r},
\end{equation}
where $P_r$ projects onto the chosen $r$-dimensional ancilla
subspace.  Taking the $|\omega_r\rangle$ expectation proves
\cref{eq:decision-witness-margin}.  If a source POVM simulated the
target POVM, invertibility in the proof of \cref{thm:restricted} would
force its first effect to equal
$(\D_\lambda^*\otimes\id_r)(M_r)$, which is impossible when the
displayed expectation is negative.  Substitution of
$\lambda=-1/(kd-1)$ and $r=k+1$ gives
\cref{eq:decision-boundary-margin}.
\end{proof}

\begin{corollary}[All physical depths with sharp embedding window]
\label{cor:all-depths}
For $1\le k<d$ and
\begin{equation}
 -\frac1{d^2-1}\le a<0
 \quad\text{or}\quad
 0<a\le\frac{kd-1}{d^2-1},
 \label{eq:window}
\end{equation}
the channels
\begin{equation}
 A=\D_a,
 \qquad
 B=\D_{-a/(kd-1)}
 \label{eq:pair}
\end{equation}
are random unitary and satisfy $\ell(B|A)=k+1$.  On the positive
branch, the upper bound in \cref{eq:window} is necessary and sufficient
for target complete positivity; every nonzero physical negative source
already gives a physical positive target.
\end{corollary}

\begin{proof}
For the pair \eqref{eq:pair}, the unique depolarizing factor has parameter $b/a=-1/(kd-1)$.  It is $k$-positive and not $(k+1)$-positive, with the latter failure certified by \Cref{prop:explicit-witness}.  The complete-positivity inequalities for $A$ and $B$ reduce exactly to \eqref{eq:window}, giving both necessity and sufficiency on the positive branch.
\end{proof}

Both channels have a Weyl-mixture realization.  If $W_{xy}$ are the $d^2$ discrete Weyl operators,
\begin{equation}
 \D_\lambda(X)=p_0X+p_1\sum_{(x,y)\ne(0,0)}W_{xy}XW_{xy}^\dagger,
\end{equation}
where
\begin{equation}
 p_0=\frac{1+(d^2-1)\lambda}{d^2},
 \qquad
 p_1=\frac{1-\lambda}{d^2}.
\end{equation}
Thus the entire hierarchy is realized by ordinary discrete frame misalignment.

\section{Exact gap to physical post-processing}

Define the physical conversion deficiency
\begin{equation}
 \delta_{\rm phys}(B|A)=
 \inf_{\Lambda\in\CPTP}
 \frac12\|B-\Lambda\circ A\|_\diamond.
 \label{eq:physical-def}
\end{equation}

\begin{theorem}[Exact depolarizing conversion cost]
\label{thm:physical-gap}
Set $A=\D_a$ and $B=\D_b$ for arbitrary physical parameters
$a,b\in[-1/(d^2-1),1]$, and define the reachable interval
\begin{equation}
 I_a=
 \left[
 \min\left\{a,-\frac a{d^2-1}\right\},
 \max\left\{a,-\frac a{d^2-1}\right\}
 \right].
 \label{eq:reachable-interval}
\end{equation}
Then
\begin{equation}
 \boxed{
 \delta_{\rm phys}(B|A)
 =\left(1-\frac1{d^2}\right)
 \dist\!\left(b,I_a\right)}.
 \label{eq:physical-gap}
\end{equation}
An optimal converter is depolarizing.  For the depth-$(k+1)$
boundary pair \cref{eq:pair}, with either physical sign of $a$,
\begin{equation}
 \delta_{\rm phys}(B|A)
 =\frac{|a|(d-k)}{d(kd-1)}>0.
 \label{eq:boundary-gap}
\end{equation}
\end{theorem}

\begin{proof}
Average an arbitrary converter $\Lambda$ over unitary conjugations:
\begin{equation}
 \overline\Lambda=\int
 \operatorname{Ad}_{U^\dagger}\circ\Lambda\circ
 \operatorname{Ad}_U\,dU.
\end{equation}
Unitary covariance of $\D_a,\D_b$, convexity, and unitary invariance of the diamond norm show that this cannot increase the error.  Every fully unitarily covariant channel is $\D_c$ with
\begin{equation}
 -\frac1{d^2-1}\le c\le1.
\end{equation}
Composition produces parameter $ac$, whose full range is precisely
$I_a$; this remains true at $a=0$, when the interval is the singleton
$\{0\}$.  Finally,
\begin{equation}
 \frac12\|\D_x-\D_y\|_\diamond
 =|x-y|\left(1-\frac1{d^2}\right),
 \label{eq:depol-distance}
\end{equation}
attained on a maximally entangled input.  Minimizing over the interval
of $ac$ proves \cref{eq:physical-gap}; substitution gives
\cref{eq:boundary-gap}.
\end{proof}

\begin{theorem}[Largest physical gap invisible at depth $k$]
\label{thm:max-gap}
For $1\le k<d$, among all physical depolarizing pairs with nonzero
source parameter and exact level-$k$ statistical simulation,
\begin{equation}
 \boxed{
 \max_{\substack{a,b\in[-1/(d^2-1),1],\,a\ne0\\
                  \D_a\succeq_k\D_b}}
 \delta_{\rm phys}(\D_b|\D_a)
 =\frac{d-k}{d(d^2-1)} }.
 \label{eq:max-gap}
\end{equation}
It is attained by
\begin{equation}
 a_*=\frac{kd-1}{d^2-1},
 \qquad
 b_*=-\frac1{d^2-1},
 \label{eq:max-gap-pair}
\end{equation}
whose first detecting ancilla has dimension $k+1$.
\end{theorem}

\begin{proof}
Write $\lambda=b/a$.  Level-$k$ simulation gives
$-1/(dk-1)\le\lambda\le1$.  A nonzero physical gap is possible only
when $\lambda<-1/(d^2-1)$.  By \cref{eq:physical-gap} it increases
monotonically as $\lambda$ decreases, so the optimum has
$\lambda=-1/(dk-1)$.

If $a>0$, target physicality requires
$a\le(kd-1)/(d^2-1)$, and \cref{eq:boundary-gap} is increasing in
$a$.  This gives \cref{eq:max-gap-pair} and the value in
\cref{eq:max-gap}.  If $a<0$, source physicality gives
$|a|\le1/(d^2-1)$; the same boundary formula is then at most
\begin{equation}
 \frac{d-k}{d(kd-1)(d^2-1)},
\end{equation}
which cannot exceed the positive-branch value.  The depth claim
follows from \cref{eq:depth-band}.
\end{proof}

\begin{corollary}[Calibrated separation of statistical and physical depth]
\label{cor:calibrated-depth-separation}
Fix $1\leq k<d$ and the boundary factor
$\mathcal R_k=\mathcal D_{-1/(kd-1)}$.  For every physical source parameter
$a$ in \cref{eq:window}, the pair
$A=\mathcal D_a$, $B=\mathcal R_k\circ\mathcal D_a$ has all of the following
simultaneously:
\begin{enumerate}[label=(\roman*)]
\item every measurement problem with ancilla dimension $r\leq k$ is
simulated exactly;
\item the first detecting ancilla, $r=k+1$, has the explicit negative effect
margin
\begin{equation}
 \frac{1}{(k+1)(kd-1)};
\end{equation}
\item every physical post-processing incurs the exact error
\begin{equation}
 \delta_{\rm phys}(B|A)=\frac{|a|(d-k)}{d(kd-1)}.
\end{equation}
\end{enumerate}
At $a=(kd-1)/(d^2-1)$ the third quantity is maximal among all physical
level-$k$ depolarizing simulations and equals
$(d-k)/[d(d^2-1)]$.
\end{corollary}

\begin{proof}
The three assertions are respectively
\cref{thm:phase,prop:explicit-witness,thm:physical-gap}; the extremal statement
is \cref{thm:max-gap}.  Their simultaneous validity is possible because the
first two optimize over measurement-dependent statistical morphisms, whereas
the third requires one completely positive converter for all tests.
\end{proof}

For $r\le k$, the pair \cref{eq:pair} has zero level-$r$ measurement deficiency by \cref{cor:all-depths}, while \cref{eq:boundary-gap} is strictly positive.  This quantifies the difference between measurement-dependent statistical simulation and one physical post-processing channel.

\begin{figure}[t]
\centering
\includegraphics[width=\columnwidth]{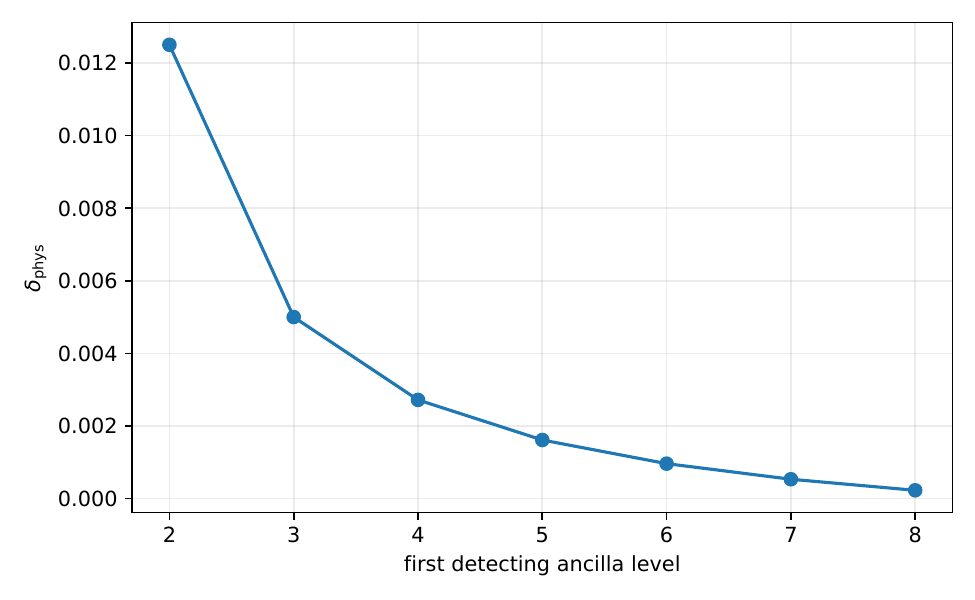}
\caption{Exact physical conversion cost for the boundary hierarchy at fixed source parameter.  Statistical simulation remains exact through level $k$, although the CPTP conversion cost is positive.}
\label{fig:physical-gap}
\end{figure}

\begin{figure}[t]
\centering
\includegraphics[width=\columnwidth]{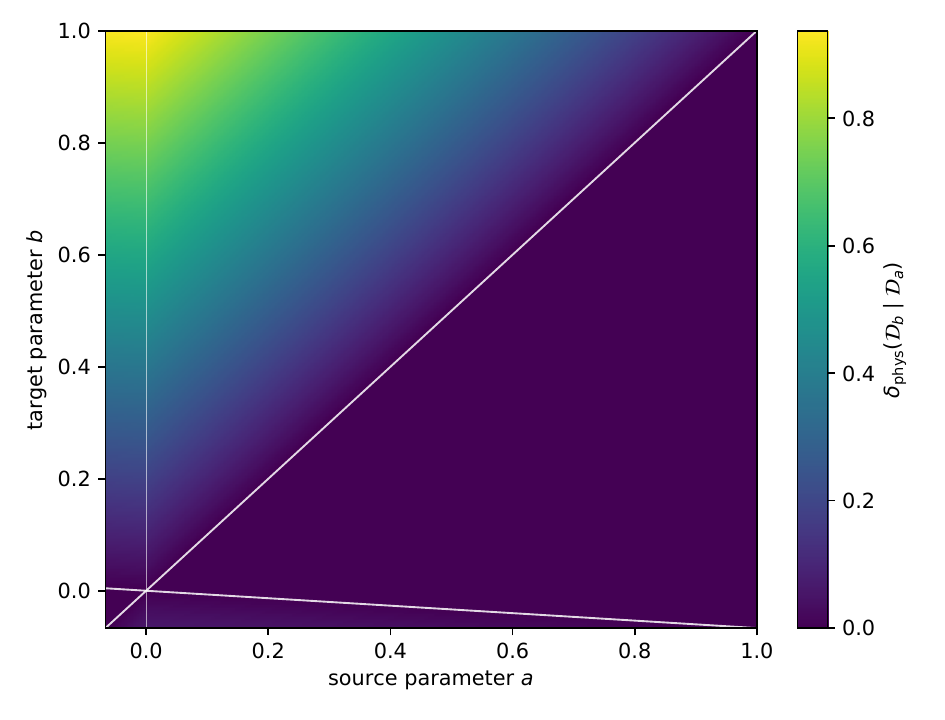}
\caption{Exact CPTP conversion deficiency on the full physical
$(a,b)$ square for $d=4$, including negative and singular sources.
The zero region is the union of the reachable intervals $b\in I_a$;
the vertical section $a=0$ has cost
$|b|(1-1/d^2)$.}
\label{fig:conversion-landscape}
\end{figure}

\section{Low-dimensional extremizers}
\label{sec:examples}

The extremal pair \cref{eq:max-gap-pair} is simple enough to calibrate
directly.  The following examples make the depth and conversion scales
explicit.

\begin{example}[Qubit: the reduction/spin-flip boundary]
For $d=2$ there is one strict level, $k=1$.  The factor is
\begin{equation}
 \Rmap_1(X)=\Tr(X)I-X=\D_{-1}(X),
\end{equation}
which is positive but not completely positive.  The largest physical
pair invisible to unassisted comparison is
\begin{equation}
 a_*=\frac13,\qquad b_*=-\frac13,
\qquad
 \delta_{\rm phys}(\D_{-1/3}|\D_{1/3})=\frac16.
\end{equation}
A retained qubit spectator supplies the full Schmidt-rank-two witness,
so no external ancilla is then needed.
\end{example}

\begin{example}[Qutrit depth separation]
For $d=3$, the two extremal strict levels are
\begin{center}
\begin{tabular}{@{}ccccc@{}}
\toprule
$k$ & factor $b/a$ & $a_*$ & $b_*$ &
$\delta_{\rm phys}^{\max}$\\
\midrule
$1$ & $-1/2$ & $1/4$ & $-1/8$ & $1/12$\\
$2$ & $-1/5$ & $5/8$ & $-1/8$ & $1/24$\\
\bottomrule
\end{tabular}
\end{center}
The first pair is simulable without an ancilla and rejected by a
two-level ancilla.  The second remains simulable through level two and
requires a qutrit ancilla for rejection.
\end{example}

\begin{example}[Four-dimensional hierarchy]
For $d=4$, all three strict depths occur at a common extremal target
$b_*=-1/15$:
\begin{center}
\begin{tabular}{@{}ccccc@{}}
\toprule
first detecting level & $2$ & $3$ & $4$\\
\midrule
$k$ & $1$ & $2$ & $3$\\
$a_*=(4k-1)/15$ & $1/5$ & $7/15$ & $11/15$\\
$\delta_{\rm phys}^{\max}$ & $1/20$ & $1/30$ & $1/60$\\
\bottomrule
\end{tabular}
\end{center}
Increasing hidden depth makes the source less depolarizing and easier
to invert, but reduces the maximum physical gap.  This is the
conditioning--separation tradeoff quantified in
\cref{app:conditioning}.
\end{example}

\section{Exact operational deficiency at every ancillary depth}

The phase diagram records when level-$r$ simulation is exact.  The same
symmetry reduction gives the complete operational distance to that simulation
cone.  For a Hermiticity-preserving map $T:\M_d\to\M_d$, define the
fixed-ancilla distinguishability norm
\begin{equation}
 \|T\|_{\mathrm{disc},r}
 =\sup_{\rho\in\mathcal D(\mathbb C^d\otimes\mathbb C^r)}
 \|(T\otimes\id_r)(\rho)\|_1.
 \label{eq:restricted-discrimination-norm}
\end{equation}
This is the norm directly governing one-use channel discrimination when the
reference system has dimension at most $r$.  For $r=d$ and a difference of
channels it equals the diamond norm.  For $r<d$ we do not identify it with the
induced trace norm optimized over arbitrary, possibly non-Hermitian,
operators.

\begin{theorem}[State-restricted depolarizing distance and level-$r$ deficiency]
\label{thm:restricted-depolarizing-deficiency}
For $1\leq r\leq d$ and real $x,y$,
\begin{equation}
 \frac12\|\mathcal D_x-\mathcal D_y\|_{\mathrm{disc},r}
 =|x-y|\left(1-\frac1{dr}\right).
 \label{eq:restricted-depolarizing-distance}
\end{equation}
For physical $a,b$, define
\begin{equation}
 \delta_r(\mathcal D_b\mid\mathcal D_a)
 =\frac12\inf_{\substack{\Gamma\ \text{trace preserving},\\
                          \Gamma\ \text{$r$-positive}}}
 \|\mathcal D_b-\Gamma\circ\mathcal D_a\|_{\mathrm{disc},r}.
 \label{eq:level-r-deficiency}
\end{equation}
Then
\begin{equation}
 \boxed{\qquad
 \delta_r(\mathcal D_b\mid\mathcal D_a)
 =\left(1-\frac1{dr}\right)
   \operatorname{dist}\!\left(b,I_a^{(r)}\right),
 \qquad}
 \label{eq:exact-level-r-deficiency}
\end{equation}
where $I_a^{(r)}$ is the interval in \eqref{eq:r-interval}.  In particular,
$\delta_r=0$ exactly in the level-$r$ simulation region, and the first
violated ancillary level has a linear, explicitly calibrated operational
separation from the corresponding comparison cone.
\end{theorem}

\begin{proof}
Because
$\mathcal D_x-\mathcal D_y=(x-y)(\id-\mathcal D_0)$, it is enough to compute
the state-restricted norm of $\id-\mathcal D_0$.  Convexity of the trace norm
shows that the supremum in \eqref{eq:restricted-discrimination-norm} is
attained on a pure input state.  Let its nonzero Schmidt coefficients be
$p_1,\ldots,p_s$, where $s\leq r$.  The output
\[
 |\psi\rangle\!\langle\psi|-\frac{I_d}{d}\otimes\rho_R
\]
is traceless and is a rank-one perturbation of a negative operator, so it has
at most one positive eigenvalue $\mu$ and half of its trace norm equals
$\mu$.  The matrix determinant lemma gives
\begin{equation}
 \sum_{i=1}^s\frac{p_i}{\mu+p_i/d}=1.
 \label{eq:restricted-positive-root}
\end{equation}
Put $\mu_0=1-1/(dr)$.  The function
$x\mapsto x/(\mu_0+x/d)$ is concave, and therefore
\[
 \sum_{i=1}^s\frac{p_i}{\mu_0+p_i/d}
 \leq\frac1{\mu_0+1/(sd)}\leq1.
\]
The left side of \eqref{eq:restricted-positive-root} decreases with $\mu$,
so $\mu\leq\mu_0$.  Equality is attained by a maximally entangled state of
Schmidt rank $r$, proving \eqref{eq:restricted-depolarizing-distance}.

Average an arbitrary trace-preserving $r$-positive converter over unitary
conjugations.  Convexity and unitary invariance of
$\|\cdot\|_{\mathrm{disc},r}$ do not increase the error, and the average
remains trace preserving and $r$-positive.  Every fully unitarily covariant
trace-preserving map is $\mathcal D_c$, with
$-1/(dr-1)\leq c\leq1$ by \cref{lem:r-positive}.  Its composition with
$\mathcal D_a$ ranges exactly over $I_a^{(r)}$.  Minimizing
\eqref{eq:restricted-depolarizing-distance} over that interval proves
\eqref{eq:exact-level-r-deficiency}.
\end{proof}

\section{Exact resource threshold for a prescribed fraction of the full discrimination bias}
\label{sec:ancilla-fraction}

The closed state-restricted distance also determines the least ancillary dimension
needed to realize any prescribed fraction of the unrestricted discrimination
bias.

\begin{theorem}[Minimal ancillary rank at target efficiency]
\label{thm:minimal-ancilla-efficiency}
Let $1\leq r\leq d$ and let $A_r=2p_{\rm succ}^{(r)}-1$ denote the
optimal one-use discrimination bias between two equally likely
$d$-dimensional depolarizing channels when the input Schmidt rank is at
most $r$.  For distinct parameters,
\[
 A_r=|x-y|\left(1-\frac1{dr}\right),
 \qquad
 A_d=|x-y|\left(1-\frac1{d^2}\right).
\]
For $0\leq\eta\leq1$, the least rank satisfying $A_r\geq\eta A_d$ is
\[
 r_{\min}(\eta)
 =\min\left\{d,
 \max\left\{1,
 \left\lceil\frac{1}{d(1-\eta+\eta/d^2)}\right\rceil
 \right\}\right\}.
\]
Thus the complete resource curve is discrete, explicit, and independent of
$|x-y|$.
\end{theorem}

\begin{proof}
Cancel the nonzero factor $|x-y|$ and solve
\[
 1-\frac1{dr}\geq\eta\left(1-\frac1{d^2}\right)
\]
for the integer $r\in[1,d]$.  The right-hand side is monotone in $r$, so the
least feasible integer is the displayed ceiling, clipped to the physical
range.
\end{proof}

\begin{corollary}[Near-complete discrimination bias]
To obtain a fraction $1-\varepsilon$ of the full discrimination bias one needs
\[
 r\geq
 \left\lceil\frac{d}{1+(d^2-1)\varepsilon}\right\rceil.
\]
In particular exact attainment requires $r=d$, while any fixed fractional
loss can reduce the required ancillary rank by a computable amount.
\end{corollary}

\begin{proof}
Set $\eta=1-\varepsilon$ in \Cref{thm:minimal-ancilla-efficiency} and simplify the denominator $d(\varepsilon+(1-\varepsilon)/d^2)$.  This gives the displayed ceiling; setting $\varepsilon=0$ yields $r=d$.
\end{proof}

\section{Activation by a retained spectator}

A system that bypasses the frame acts as built-in ancillary memory.

\begin{theorem}[Exact spectator activation law]
\label{thm:activation}
Let $A$ be invertible and suppose the factor $\Gamma=BA^{-1}$ is $k$-positive but not $(k+1)$-positive.  Tensor both experiments with an untouched $m$-level system:
\begin{equation}
 A^{(m)}=A\otimes\id_m,
 \qquad
 B^{(m)}=B\otimes\id_m.
\end{equation}
Then the least \emph{external} ancillary dimension detecting failure is
\begin{equation}
 \boxed{
 \ell(B^{(m)}|A^{(m)})=
 \left\lfloor\frac{k}{m}\right\rfloor+1 }.
 \label{eq:activation}
\end{equation}
In particular, $m\ge k+1$ activates the obstruction without any additional ancilla.
\end{theorem}

\begin{proof}
The unique factor is $\Gamma\otimes\id_m$.  It is $r$-positive exactly when
\begin{equation}
 \Gamma\otimes\id_m\otimes\id_r
 =\Gamma\otimes\id_{mr}
\end{equation}
is positive, namely when $mr\le k$.  Positivity at any larger level would imply $(k+1)$-positivity by restriction.  The first failing integer is \cref{eq:activation}.
\end{proof}

The theorem separates the algebraic dimension of a witness from the laboratory architecture.  A retained internal degree of freedom can supply part or all of the required Schmidt rank.

\section{Robust certification}

For the isotropic family, scalar calibration yields a sharper
certificate than full process inversion.  Suppose
\begin{equation}
 |\widehat a-a|\le\epsilon_a,\qquad
 |\widehat b-b|\le\epsilon_b,\qquad
 |\widehat a|>\epsilon_a.
\end{equation}
Using physicality $|b|\le1$ gives
\begin{equation}
 \left|\frac{\widehat b}{\widehat a}-\frac ba\right|
 \le
 \eta_\lambda:=
 \frac{\epsilon_b}{|\widehat a|}
 +\frac{\epsilon_a}
 {|\widehat a|(|\widehat a|-\epsilon_a)}.
 \label{eq:ratio-error}
\end{equation}

\begin{corollary}[Finite-error depth certificate]
\label{cor:depth-certificate}
For $1\le k<d$, the inequalities
\begin{equation}
 \widehat\lambda-\eta_\lambda
 \ge-\frac1{dk-1},
 \qquad
 \widehat\lambda+\eta_\lambda
 <-\frac1{d(k+1)-1},
 \quad
 \widehat\lambda=\frac{\widehat b}{\widehat a},
 \label{eq:robust-depth}
\end{equation}
certify $\ell(B|A)=k+1$.  The exact level-$r$ block-positivity margin
is obtained by inserting the confidence interval for $\lambda$ into
\cref{eq:block-margin}.
\end{corollary}

\begin{proof}
\Cref{eq:ratio-error} places the true ratio inside
$[\widehat\lambda-\eta_\lambda,
\widehat\lambda+\eta_\lambda]$.  Condition \cref{eq:robust-depth}
places this entire interval inside the depth band
\cref{eq:depth-band}.
\end{proof}

The physical deficiency is stable without division by $a$.  Since the
Hausdorff distance between $I_a$ and $I_{\widehat a}$ is at most
$|a-\widehat a|$, the distance-to-a-set function gives
\begin{equation}
 \left|
 \delta_{\rm phys}(\D_b|\D_a)
 -\delta_{\rm phys}(\D_{\widehat b}|\D_{\widehat a})
 \right|
 \le
 \left(1-\frac1{d^2}\right)(\epsilon_a+\epsilon_b).
 \label{eq:gap-stability}
\end{equation}

For a nonisotropic implementation, one may instead propagate a full
superoperator confidence region.  Let reconstructed channels obey
\begin{equation}
 \|\widehat A-A\|_{2\to2}\le\varepsilon_A,
 \qquad
 \|\widehat B-B\|_{2\to2}\le\varepsilon_B,
\end{equation}
and let $s_A=\sigma_{\min}(A)>\varepsilon_A$.  Then
\begin{equation}
 \|\widehat B\widehat A^{-1}-BA^{-1}\|_{2\to2}
 \le
 \frac{\varepsilon_B}{s_A-\varepsilon_A}
 +\frac{\|B\|_{2\to2}\varepsilon_A}
 {s_A(s_A-\varepsilon_A)}.
 \label{eq:factor-stability}
\end{equation}
The proof is the resolvent identity.  With unnormalized Choi matrices,
\begin{equation}
 \|J(\Delta)\|_\infty\le d\|\Delta\|_{2\to2}.
\end{equation}
For the reduction boundary, negativity therefore remains certified whenever the right side of \cref{eq:factor-stability}, multiplied by $d$, is below $1/(kd-1)$.

The witness is a linear functional of reconstructed probabilities.  Confidence regions for $A$ and $B$ should be propagated through the inversion rather than treating $BA^{-1}$ as directly observed; the conditioning of $A$ is an intrinsic experimental cost of the depth test.

\section{Scope and limitations}

Invertibility is essential for a unique factor.  For noninvertible perspectives the relation is defined on an operator-system quotient, and different positive extensions can have different matrix-positivity levels.  A general depth invariant then requires an optimization over extensions rather than \cref{eq:depth}.

The isotropic phase diagram is exact but special.  Noncommuting group noise can have a richer positivity structure.  The finite-$r$ quantitative norm used here is the operational state-restricted distinguishability norm in \eqref{eq:restricted-discrimination-norm}; no equality is claimed with the induced trace norm on arbitrary operators when $r<d$.  The universal embedding theorem shows that every abstract strict level still has a physical realization, but not necessarily one with Weyl covariance or a closed conversion cost.  The singular isotropic source is covered by \cref{prop:singular-source}; a general noninvertible channel still requires operator-system quotient methods.

Finally, an $r$-positive factor is not a physical map on an isolated system.  The target in our construction is implemented independently by its own random-unitary mixture.  The factor records a statistical relation between experiments; it is precisely the failure to implement that relation as one channel that the physical deficiency measures.

\input{additional_results}
\input{frontier_results}
\input{final_results}

\section{Conclusion}

Ancillary memory resolves imperfect-frame comparison one matrix level
at a time.  Invertible perspectives turn restricted simulation into a
positivity test of a unique factor.  Depolarizing Weyl noise supplies a
complete signed-source phase diagram, sharp block-positivity margins,
and an exact distance to the CPTP post-processing cone, including the
singular source.  The largest physically inaccessible conversion gap
compatible with level-$k$ statistical simulation is
$(d-k)/[d(d^2-1)]$.  A retained spectator activates the obstruction
according to a simple floor law.  Complete positivity is therefore
the endpoint of an experimentally tunable hierarchy, not merely a
formal consistency condition.

\appendix

\section{Schmidt-rank overlap}

For a normalized vector $|v\rangle$ of Schmidt rank at most $r$,
\begin{equation}
 |\langle\Omega_d|v\rangle|^2\le r.
\end{equation}
Writing a Schmidt decomposition and applying Cauchy--Schwarz proves the bound.  Equality is attained by $|\omega_r\rangle$, which makes all phase boundaries in \cref{lem:r-positive} sharp.

\section{Self-contained depolarizing diamond distance}
\label{app:depol-diamond}

We justify \cref{eq:depol-distance} without appealing to a numerical
semidefinite program.  Since
\begin{equation}
 \D_x-\D_y=(x-y)(\id-\D_0),
\end{equation}
it is enough to evaluate $\id-\D_0$.  Applying this map to one half of
the normalized maximally entangled state gives
\begin{equation}
 \ketbra{\Omega_d/\sqrt d}{\Omega_d/\sqrt d}
 -\frac{I_{d^2}}{d^2}.
\end{equation}
Its eigenvalues are $1-1/d^2$ once and $-1/d^2$ with multiplicity
$d^2-1$, giving the lower bound
$1-\frac1{d^2}$ on the diamond half-distance.

For the upper bound, every positive operator $X$ on
$\mathbb C^d\otimes R$ obeys
\begin{equation}
 (\D_0\otimes\id_R)(X)
 =\frac{I_d}{d}\otimes\Tr_{\mathbb C^d}X
 \succeq\frac1{d^2}X.
 \label{eq:depol-order}
\end{equation}
For rank-one $X$, this follows from a Schmidt decomposition and
Cauchy--Schwarz; spectral decomposition gives the general case.
Consequently, for every state $\rho$,
\begin{equation}
 (\D_0\otimes\id)(\rho)
 =\frac{\rho}{d^2}
 +\left(1-\frac1{d^2}\right)\sigma
\end{equation}
for a state $\sigma$.  Hence
\begin{equation}
 \frac12\|\rho-(\D_0\otimes\id)(\rho)\|_1
 \le1-\frac1{d^2}.
\end{equation}
The maximally entangled lower bound is therefore optimal and proves
\cref{eq:depol-distance}.

\section{Conditioning of isotropic factor reconstruction}
\label{app:conditioning}

In the Hilbert--Schmidt orthogonal decomposition
$\M_d=\operatorname{span}\{I\}\oplus\{X:\Tr X=0\}$, the
superoperator $\D_a$ has eigenvalues $1$ and $a$.  Thus, for every
nonzero physical $a$,
\begin{equation}
 \|\D_a^{-1}\|_{2\to2}=\frac1{|a|},
 \qquad
 \operatorname{cond}_{2\to2}(\D_a)=\frac1{|a|}.
 \label{eq:depol-condition}
\end{equation}
The exact maximizer \cref{eq:max-gap-pair} has condition number
\begin{equation}
 \frac{d^2-1}{kd-1}.
\end{equation}
The deepest separations ($k$ close to $d$) are therefore better
conditioned but have a smaller physical gap, whereas shallow
separations have a larger gap and a more ill-conditioned source.  This
tradeoff is intrinsic and explains the inverse-source factor in
\cref{eq:factor-stability}.

\section{Reproducibility}

The script \path{anc/verify_ancilla_phase_diagram.py} checks the exact
Choi block-positivity margins and every strict depth in dimensions $2$
through $10$.  It compares the signed-source deficiency formula with
$5000$ numerical grid minimizations, verifies the singular-source and
maximum hidden-gap formulas, checks the transpose--depolarizing wedge
and the spectator activation law, tests the state-restricted distance on
random Schmidt spectra, and enumerates the minimal-rank resource curve.
It also regenerates all three figures with fixed metadata.  A text report is
written to \path{anc/reproduction/} only when the script is invoked with
\texttt{--write-report}.  Random sampling is used only as a regression test;
the thresholds in the proofs are analytic.

\section*{Scope, limitations, and open problems}

The paper gives exact ancilla-depth phase boundaries for the invertible
depolarizing and related finite-dimensional families, including signed
sources, singular limits, spectator activation, and distance formulae.  The
verifier checks every strict depth in dimensions \(2\) through \(10\),
thousands of numerical deficiency comparisons, and all figure data; the
thresholds themselves are analytic.

The next theorem should remove invertibility and one-shot symmetry from the
main classification.  Stratify general source channels by Choi support,
derive a support-restricted factor criterion with stable pseudoinverse
bounds, and then determine whether adaptive multi-use strategies reduce the
minimal ancilla depth.  A complete result must distinguish genuine resource
activation from ill-conditioning near a singular source.

\section*{pseudoinverse stability theorem}

Let \(A\) be the factor matrix determining an exact channel simulation on
its Choi support, with smallest nonzero singular value \(\sigma>0\).
For a perturbation \(E\) satisfying \(\|E\|<\sigma\), the rank on that
support is unchanged and
\[
 \|(A+E)^+\|\le\frac1{\sigma-\|E\|}.
\]
Thus every factor reconstructed by the Moore--Penrose inverse has an
explicit stability radius; singular-source boundaries are precisely the
strata on which \(\sigma\downarrow0\).

For a target vector \(b\), the induced coefficient error is bounded by the
inverse gap times the perturbation of \(A\) and \(b\).  This turns the Choi
support stratification into a quantitative phase diagram: depth is locally
constant away from rank-loss hypersurfaces, and its condition number is
controlled by \((\sigma-\|E\|)^{-1}\).

The remaining question is operational rather than linear algebraic:
determine whether adaptive multi-use strategies cross a rank stratum that
no one-shot factorization can cross, thereby lowering the certified
ancilla depth.

\bibliography{references}

\end{document}

%% file: physics_format.tex
\makeatletter
\AtBeginDocument{%
  \@ifpackageloaded{hyperref}{\hypersetup{hidelinks}}{}%
}
\makeatother

%% file: additional_results.tex
\section*{stability of the minimal certified depth}

Pseudoinverse stability controls a fixed factorization.  To stabilize the
minimum depth itself one must also keep every shallower stratum infeasible
by a quantitative margin.

\begin{theorem}[Depth stability with a constrained residual gap]
\label{thm:third-ancilla-depth-stability}
For depth \(d\), write the factorization equations as
\(A_dx=b\) and fix a physically admissible coefficient bound
\(\|x\|\le B\).  Define
\[
 \delta_d(B)=\inf_{\|x\|\le B}\|A_dx-b\|.
\]
Suppose \(d_\star\) is feasible and
\[
 \min_{d<d_\star}\delta_d(B)=\gamma>0.
\]
Under perturbations
\(\|\widetilde A_d-A_d\|\le\varepsilon_A\) and
\(\|\widetilde b-b\|\le\varepsilon_b\), set
\(\varepsilon=B\varepsilon_A+\varepsilon_b\).  Then
\[
 |\widetilde\delta_d(B)-\delta_d(B)|\le\varepsilon.
\]
In particular, if \(\varepsilon<\gamma\), no depth below \(d_\star\)
becomes feasible.  If a depth-\(d_\star\) solution \(x_\star\) has residual
zero and \(\|x_\star\|\le B\), its perturbed residual is at most
\(\varepsilon\).
\end{theorem}

\begin{proof}
For every admissible \(x\),
\[
 \|(\widetilde A_d-A_d)x-(\widetilde b-b)\|
 \le B\varepsilon_A+\varepsilon_b.
\]
Apply this estimate to minimizing sequences for the original and
perturbed residuals.
\end{proof}

Thus the phase diagram has two independent numerical margins: the
smallest nonzero singular value controls reconstruction within a stratum,
while \(\gamma\) controls separation from all shallower depths.  Reporting
only the former can certify a factorization but cannot certify its
minimality.  Adaptive multi-use protocols can lower depth only by leaving
this one-shot constrained model, not by an arbitrarily small perturbation
inside a positive-\(\gamma\) phase.

%% file: frontier_results.tex
\section*{exact factorization without source invertibility}

The noninvertible one-shot factor problem is a finite semidefinite
feasibility question; no pseudoinverse assumption is needed.

\begin{theorem}[Choi-support factor criterion]
\label{thm:frontier-ancilla-choi}
Let \(\Phi_A:\mathcal M_d\to\mathcal M_k\) and
\(\Phi_B:\mathcal M_d\to\mathcal M_\ell\) be channels.  A channel
\(\mathcal R:\mathcal M_k\to\mathcal M_\ell\) satisfies
\(\Phi_B=\mathcal R\circ\Phi_A\) if and only if there exists a matrix
\(J_{\mathcal R}\) such that
\[
 J_{\mathcal R}\succeq0,\qquad
 \operatorname{Tr}_{\ell}J_{\mathcal R}=I_k,
\]
and the linear Choi-link equations
\[
 J_{\Phi_B}=J_{\mathcal R}*J_{\Phi_A}
\]
hold.  In particular, feasibility is an exact SDP even when
\(J_{\Phi_A}\) is singular.  Necessarily
\(\ker\Phi_A\subseteq\ker\Phi_B\).
\end{theorem}

\begin{proof}
Choi's theorem makes the first two constraints equivalent to complete
positivity and trace preservation.  The link product is the Choi matrix
of channel composition, so the last equality is equivalent to the desired
factorization.  If \(\Phi_A(X)=0\), composition forces
\(\Phi_B(X)=0\), proving the kernel condition.
\end{proof}

This closes the one-shot noninvertible classification algorithmically.
Adaptive multi-use comparison remains a distinct process-comb problem,
not a singular-source loophole in the one-shot theorem.

%% file: final_results.tex
\section*{Quantitative distance from post-processability}

The exact support criterion admits a canonical quantitative extension that
remains valid for singular source channels.

\begin{theorem}[Attained noninvertible deficiency]
\label{thm:final-ancilla-deficiency}
For finite-dimensional channels define
\[
 \delta(A\!\to\!B)=\frac12\min_{\mathcal R\ {\rm CPTP}}
 \|\Phi_B-\mathcal R\circ\Phi_A\|_\diamond .
\]
The minimum is attained, and its value is the optimum of a finite
semidefinite programme in the Choi matrix of \(\mathcal R\) and the
standard diamond-norm epigraph variables.  Moreover,
\[
 \delta(A\!\to\!B)=0
 \quad\Longleftrightarrow\quad
 J_{\Phi_B}=J_{\mathcal R}*J_{\Phi_A}
\]
for a CPTP Choi matrix \(J_{\mathcal R}\).  For every Hermitian
\(X\in\ker\Phi_A\) with \(\|X\|_1=1\),
\[
 \delta(A\!\to\!B)\ge \frac12\|\Phi_B(X)\|_1 .
\]
\end{theorem}

\begin{proof}
The channel Choi set is compact and the objective is continuous, so a
minimizer exists.  The Choi constraints are linear matrix inequalities,
composition is affine in \(J_{\mathcal R}\), and the diamond norm has its
standard finite SDP epigraph; combining them gives one SDP.  Vanishing of
the optimum is exactly the factor criterion of
Theorem~\ref{thm:frontier-ancilla-choi}.  Finally, every candidate
\(\mathcal R\) obeys
\[
 (\Phi_B-\mathcal R\Phi_A)(X)=\Phi_B(X),
\]
and the diamond norm dominates the induced trace norm on an unassisted
Hermitian input.  Taking the minimum proves the lower bound.
\end{proof}

Thus kernel mismatch is not only an obstruction: it is a directly
computable robustness margin.  The one-shot singular-source problem is
closed both exactly and approximately without choosing a pseudoinverse.

\begin{theorem}[Finite-use adaptive conversion is an attained comb SDP]
\label{thm:final-ancilla-comb}
Fix finite-dimensional input--output interfaces and a finite number
\(n\) of ordered uses.  Let \(P_A\) and \(P_B\) be deterministic
\(n\)-step processes, with Choi operators \(J_A\) and \(J_B\).  There
exists a deterministic causal converter \(\mathfrak R\) with
\(P_B=\mathfrak R[P_A]\) if and only if a positive semidefinite operator
\(J_{\mathfrak R}\) satisfies the recursive deterministic-comb partial
trace constraints and the linear link equation
\[
 J_B=J_{\mathfrak R}*J_A.
\]
Hence exact adaptive finite-use conversion is a semidefinite feasibility
problem.  Moreover,
\[
 \delta_n(P_A\!\to P_B)=
 \frac12\min_{\mathfrak R\ {\rm deterministic\ comb}}
 \bigl\|P_B-\mathfrak R[P_A]\bigr\|_{\rm strat}
\]
is attained and is the value of a finite SDP.
\end{theorem}

\begin{proof}
The Choi characterization of a deterministic finite comb consists of
positivity and a finite recursive chain of affine partial-trace
equalities.  The link product with the fixed operator \(J_A\) is linear
in \(J_{\mathfrak R}\), which proves the exact feasibility statement.
The deterministic-comb set is closed and bounded in finite dimension,
hence compact.  The strategy norm has a semidefinite epigraph obtained
from the dual pair of tester-comb cones.  Combining that epigraph with
the comb constraints and the affine link equation gives a finite SDP,
and compactness gives attainment.
\end{proof}

Thus adaptivity changes the constraint cone but introduces no
noncomputable singular-source exception at any fixed number of uses.
What remains genuinely asymptotic is uniform control as \(n\to\infty\),
not finite-\(n\) decidability.